\documentclass[runningheads, a4paper]{llncs}

\usepackage{amssymb}
\usepackage{graphicx}
\usepackage{subfigure}
\usepackage{url}
\urldef{\mailsa}\path|{fei125, williamhao}@sjtu.edu.cn, {gao-xf, yaobin,gchen}@cs.sjtu.edu.cn|

\newcommand{\keywords}[1]{\par\addvspace\baselineskip
\noindent\keywordname\enspace\ignorespaces#1}

\usepackage{float}
\usepackage{cite}

\usepackage{subfigure}
\usepackage[usenames,dvipsnames]{xcolor}

\usepackage[cmex10]{amsmath}
\usepackage{dsfont}
\usepackage{bbm}
\usepackage{bm}
\usepackage[misc]{ifsym}
\usepackage{mathrsfs}
\usepackage{bbding}

\usepackage[vlined,ruled,linesnumbered]{algorithm2e}

\begin{document}

\setlength{\algoheightrule}{0.8pt}

\setlength{\intextsep}{1em}

\mainmatter  

\title{QDR-Tree: An Efficient Index Scheme for Complex Spatial Keyword Query}

\author{Xinshi Zang\and Peiwen Hao\and Xiaofeng Gao\inst{(}\Envelope\inst{)}\and Bin Yao \and Guihai Chen}
\authorrunning{X. Zang et al.}

\institute{Department of Computer Science and Engineering, \\ Shanghai Jiao Tong University, China\\ \mailsa}

\maketitle

\begin{abstract}
	
	With the popularity of mobile devices and the development of geo-positioning technology, location-based services (LBS) attract much attention and top-$k$ spatial keyword queries become increasingly complex.
	It is common to see that clients issue a query to find a restaurant serving pizza and steak, low in price and noise level particularly.
	However, most of prior works focused only on the spatial keyword while ignoring these independent numerical attributes.
	
	\qquad In this paper we demonstrate, for the first time, the \emph{Attributes-Aware Spatial Keyword Query} (ASKQ), and devise a two-layer hybrid index structure called \emph{Quad-cluster Dual-filtering R-Tree} (QDR-Tree). In the keyword cluster layer, a Quad-Cluster Tree (QC-Tree) is built based on the hierarchical clustering algorithm using kernel $k$-means to classify keywords.
	In the spatial layer, for each leaf node of the QC-Tree, we attach a Dual-Filtering R-Tree (DR-Tree) with two filtering algorithms, namely, keyword bitmap-based and attributes skyline-based filtering. Accordingly, efficient query processing algorithms are proposed.
	
	\qquad Through theoretical analysis, we have verified the optimization both in processing time and space consumption. Finally, massive experiments with real-data demonstrate the efficiency and effectiveness of QDR-Tree.
	
	
\end{abstract}

\section{Introduction}\label{sec-intr}
With the growing popularity of mobile devices and the advance in geo-positioning technology, location-based services (LBS) are widely used and spatial keyword query becomes increasingly complex.
Clients may have special requests on numerical attributes, such as price, in addition to the location and keywords.

\begin{example} Consider some spatial objects in Fig.~\ref{fig:basic_exam:a}, where dots represent spatial objects such as restaurants, whose keywords and three numerical attributes are listed in Fig. \ref{fig:basic_exam:b}. Dots with the same color own similar keywords, e.g., red dots share keywords about food. The triangle represents a user issuing a query to find a nearest restaurant serving pizza and steak with low level in price, noise, and congestion. At a first glance, $o_8$ seems to be the best choice for the close range, while $o_1$ surpasses $o_8$ in the numerical attributes obviously.
	This common situation shows that such complex queries deserve careful treatment.
	
	\label{exa:1}
\end{example}

\begin{figure}
	\centering
	\subfigure[Geo-position of spatial objects]{
		\label{fig:basic_exam:a} 
		\includegraphics[width=1.9in]{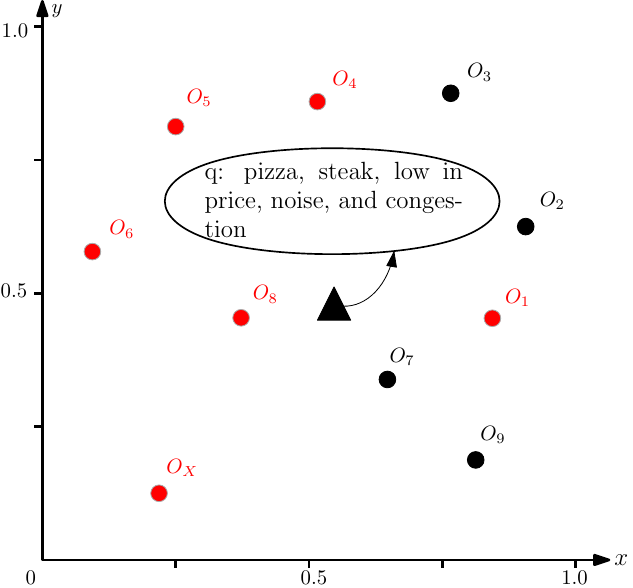}}
	\hspace{0.2in}
	\subfigure[Keywords and attributes table]{
		\label{fig:basic_exam:b} 
		\includegraphics[width=2.3in]{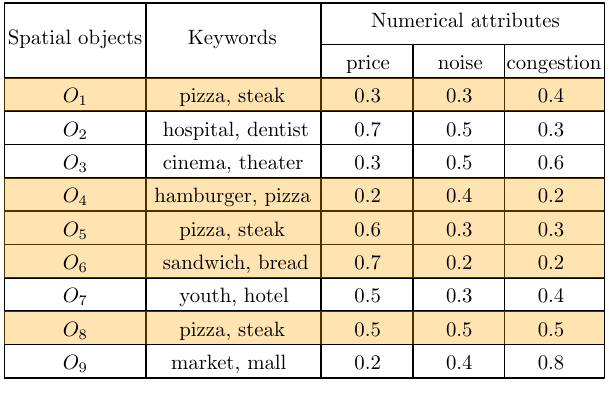}}
	\caption{A set of spatial objects and a query}
	\label{fig:basic_exam} 
\end{figure}

Extensive efforts have been made to support spatial keyword query.
However, prior works \cite{de2008keyword,li2011ir,ray2016dynamically} mainly
focused on the keywords of spatial objects but neglected or failed to distinguish independent numerical attributes.
Recently, Sasaki Y. \cite{sasaki2014sky} schemed out SKY R-Tree which incorporates R-tree with skyline algorithm to deal with the numerical attributes.
However, it does not work well for multi keywords, which reduces their usage for various applications.
Liu X.~\cite{6940236} proposed a hybrid index structure called Inverted R-tree with Synopses tree (IRS), which can search many different types of numerical attributes simultaneously.
However, the IRS-based search algorithm requires  providing exact ranges of attributes which is a heavy and unnecessary burden to the users.
What's more, the exact match in in attributes can also lead to  few or no query results to be returned.


Correspondingly, in this paper, we named and studied, for the first time,
the \emph{attributes-aware spatial keyword query} (ASKQ). This complex query needs to take
location proximity, keywords' similarity, and the value of numerical attributes into
consideration, that is respectively, the Euclidean spatial distance, the relevance of different keywords,
and the integrated attributes of users' preference.
Obviously the ASKQ has wide apps in the real world.

Tackling with the ASKQ in Example \ref{exa:1}, common search algorithms \cite{de2008keyword,li2011ir,ray2016dynamically}
ignoring numerical attributes may retrieve finally $o_1, o_5, o_8$ indiscriminately,
and SKY R-Tree-based algorithm may return $o_4$ as one of results, and IRS-Tree-based algorithm may retrieve no objects when the query predicate is set as ``$price < 0.3$ \& $noise<0.3$ \& $congestion<0.4$''. Apparently, none of these algorithms can satisfy the users' need.
These gaps motivate us to investigate new approaches that can deal with the ASKQ efficiently.

In this paper, we propose a novel two-layer index structure called \emph{Quad-cluster Dual-filter  R-Tree} (QDR-Tree) with query processing algorithms. In the first layer we deal with keyword specifically. Considering numbers of keywords share the similar semanteme and clients tend to query objects of the same class, we cluster and store the keywords in a Quad-Cluster Tree (QC-Tree) by hierarchical clustering algorithm using kernel $k$-means clustering~\cite{Dhillon2004Kernel}. With keyword relaxation operation and Cut-line theorem to avoid redundance, QC-Tree can balance search time and space cost well.

In the second layer we deal with spatial objects with numerical attributes. At each leaf node of the first layer, a Dual-filter R-Tree (DR-Tree) is attached according to two filtering algorithms, namely, keyword bitmap-based filtering and attributes skyline-based filtering, which effectively reduce the false positives. 


Moreover, we also propose a novel method to measure the relevance of one spatial object with the query keywords. We measure the similarity of different keywords from both textual and semantic aspects. For the latter one, the \emph{term vectors} that are obtained by word2vec
\cite{mikolov2013efficient} are applied to represent every keywords, and therefore, the similarity can be quantified.
Note that both queries and spatial objects usually own several keywords,
a bitmap of keywords is used to measure the relevance between two lists of keywords lightly and efficiently.


Table~\ref{tab:1} compares the current index with QDR-Tree in three aspects.
Apparently, QDR-Tree outperform existing methods in tackling with the ASKQ,
and can achieve great improvements in query processing time and space consumption.
This  will be demonstrated in both theoretical and experimental analysis.
Massive experiments with real-data also confirm the efficiency of QDR-Tree.

\begin{table}
	\centering
	\caption{Comparisons among current indexes and QDR-Tree}
	\label{tab:1}
	\begin{tabular}{ccccc}
		\hline
		Index& From & location proximity & muti-keywords& fuzzy attributes \\
		\hline
		IR-Tree & TKDE(2011)\cite{li2011ir}& \Checkmark & \Checkmark& \XSolidBrush \\
		
		IL-Quadtree& ICDE(2013)\cite{Zhang2013Inverted}&\Checkmark & \Checkmark& \XSolidBrush \\
		
		SKY R-Tree & DASFAA(2014)\cite{sasaki2014sky}&\Checkmark & \XSolidBrush &\Checkmark\\
		
		IRS-Tree& TKDE(2015) \cite{6940236}&\Checkmark & \Checkmark& \XSolidBrush\\
		
		QDR-Tree & DEXA(2018) &\Checkmark &\Checkmark &\Checkmark \\
		\hline
	\end{tabular}
\end{table}

To sum up, the main contributions of this paper are summarized as follows:
\begin{itemize}
	\item We formulate the attributes-aware spatial keyword query, which takes spatial proximity, keywords' similarity and numerical attributes into consideration.
	
	\item We design a novel hybrid index structure, i.e., QDR-Tree which incorporates Quad-Cluster Tree with Dual-filtering R-Trees and accordingly propose the query processing algorithm to tackle the ASKQ.
	\item We propose a novel method to measure the relevance of one spatial objects with query keywords based on word2vec and bitmap of keyword.
	\item We conduct an empirical study that demonstrates the efficiency of our algorithms and index structures for processing the ASKQ on real-world datasets.
\end{itemize}

The rest of the paper is organized as follows. Section \ref{related} reviews the related works.
Section \ref{formulation} formulates the problem of ASKQ.
Section \ref{SCIR} presents the QDR-Tree. Section \ref{enhanced} introduces the query processing algorithm based on the QDR-Tree.
Three baseline algorithms are proposed in Section \ref{performance} and
considerable experimental results are reported. Finally,  Section \ref{conclusion} concludes the paper.

\section{Related Work} \label{related}
Existing works concerning the ASKQ include spatial keyword search, keyword relevance measurement, and the skyline operator.

\textbf{Spatial keyword search.}
There are many studies on spatial keyword search recently \cite{de2008keyword,Zhang2013Inverted,tao2014fast}. Most of them focus on integrating inverted index
and R-tree to support spatial keyword search. For example,
IR2-tree\cite{de2008keyword} combines R-trees with signature files. It preserves
objects spatial proximity, which is the key to solve spatial
queries efficiently, and can filter a considerable portion
of the objects that do not contain all the query keywords.
Thus it significantly reduces the number of objects to be examinated.
SI-index \cite{Zhang2013Inverted} overcomes IR2-trees' drawbacks and 
outperform IR2-tree in query response time significantly. \cite{tao2014fast}
proposes inverted linear quadtree, which is carefully designed to exploit both
spatial and keyword-based pruning techniques to effectively
reduce the search space.

\textbf{Keyword relevance measurement.}
The traditional measurement on keyword relevance includes textual and semantic relevance. The textual relevance can be computed using an information retrieval model \cite{Cong2009Efficient,Cong2009337,zobel2006inverted,Cao2010Retrieving}. They are all TF-IDF variants essentially
sharing the same fundamental principles.
The semantic relevance is measured by many methods.  \cite{Qian2016On, Qian2017Semantic} apply the Latent Dirichlet Allocation (LDA) model to calculate the topic distance of keywords. Gao Y.\cite{Lu2017Efficient} proposed an efficient disk-based metric access method which achieves excellent performance in the measurement of keywords' similarity.

\textbf{The skyline operator.}
The skyline operator deals with the optimization problem of selecting multi-dimension points.
A skyline query returns a set of points that are not dominated by any other
points, called a skyline. It is said that a point $o_i$ dominates
another point $o_j$ if $o_i$ is no worse than $o_j$ in all dimensions of attributes and
is better than $o_j$ at least in one dimension.
Borzsonyi  et al. \cite{borzsony2001the} first introduced the skyline operator into relational database
systems and introduced three algorithms.
Geng et al.\cite{Ma2012A} propose a method which combines the spatial information with non-spatial information to obtain skyline results.
Lee \cite{lee2014toward} et al. focused on two methods
about multi-dimensional subspace skyline computation and developed orthogonal optimization principles.

\section{Problem Statement} \label{formulation}
Given an geo-object dataset $O$ in which each object $o$ is denoted as
a tuple $\langle$$\lambda,$ K, A$\rangle$, where $o.\lambda$ is a location descriptor which we assume is at a two dimensional geographical space and is composed of latitude and longitude, $o.K$ is the set of keywords, and $o.A$ represents
the set of numerical attributes. Without loss of generality, we assume the attributes $o.a_i$ in $o.A$ are numeric attributes and normalize each $o.a_i\in [0, 1]$.
We assume that smaller values of these numercial attributes, e.g., price and noise, are preferable.
As for other numerical attributes' values which are better if higher, such as the rating and health score, we convert them decreasingly as $o.a_i = 1 - o.a_i$.
The query $q$ is represented as a tuple $\langle \lambda,K,W\rangle$, where $q.\lambda$ and
$q.K$ represent the location of the user and the required keywords respectively, and $q.W$ represents the set of weight for different numerical attributes and user's different preference on these attributes. $\forall q.w_i \in q.W, q.w_i \geq 0$ $(i = 1,\dots, |q.W|)$ and $\sum_{i=1}^{|q.W|} q.w_i = 1$.
The reason for assigning weight to each  attribute instead of qualifying exact range of attributes is to prepare for the fuzzy query on numerical attributes.
In order to elaborate  the QDR-Tree , we firstly define the keyword distance and the keyword cluster as follows.

\begin{definition}[Keyword Distance]\label{def:1}
	Given two keywords $k_1, k_2$,
	their keyword distance, denoted as $d(k_1,k_2)$, includes both textual distance and semantic distance.
	The textual similarity between two keywords is denoted as $d_t(k_1, k_2)$ which is measured by the Edit Distance.
	The semantic distance between two keywords denoted as $d_s$ is measured by the Euclidean distance of the term vector generated by word2vec.
	With a parameter $\delta(\in [0,1])$ controlling their relative weights, Eqn.~\eqref{eq:key} describes the formulation of $d(k_1, k_2)$.
	\begin{equation} \label{eq:key}
	d(k_1, k_2) = \delta d_t(k_1, k_2) + (1-\delta)d_s(k_1, k_2)
	\end{equation}
\end{definition}


\begin{definition}[Keyword Cluster]\label{def:2}
	A keyword cluster ($C_i$) is formed by similar keywords. The cluster diameter is defined as the maximum keyword distance within the cluster.
	One keyword can be allocated into the cluster if the diameter after adding it does not exceed the threshold $\tau$, i.e. $\forall k_i, k_j \in C_i, d(k_i, k_j) < \tau$. Each cluster has a center object denoted as $C_i.cen$. All the keyword clusters ($C_i$) make up the set of keyword clusters ($\mathbb{C}$).
\end{definition}

%

%
%

\begin{definition}[Attributes-Aware Spatial Keyword Query] \label{def3}
	Given a geo-object set $O$ and the attributes-aware spatial keyword query q, the result includes a set of $Top_\kappa(q)$,\footnote{Hereafter, Top-$k$ is denoted as Top-$\kappa$ to avoid confusion with the $k$-means algorithm.} $Top_\kappa(q) \subset O$, $|Top_\kappa(q)| = \kappa$ and $\forall o_i, o_j: o_i \in Top_\kappa(q), o_j \in O - Top_\kappa(q)$, it holds that $score(q, o_i) \leq score(q, o_j)$.
	
\end{definition}
As for the evaluation function, $score(q, o)$ in Def. \ref{def3}, it is composed of three aspects, including the location proximity, the keywords similarity, and the value of numerical attributes, and will be discussed at large in the Sec. \ref{enhanced}.
%
%
%

\section{QDR-Tree}\label{SCIR}
In this section, we introduce a new hybrid index structure QDR-Tree, which is a new indexing framework for efficiently processing the ASKQ. The QDR-Tree can be divided into two layers,  the keyword cluster layer and the spatial layer where the QDR-Tree can be split up into two sub-trees, named as Quad-Cluster Tree (QC-Tree) and Dual-filtering R-tree (DR-Tree) respectively.

\subsection{Keyword cluster layer}

The keyword cluster layer deals with keyword search with both textual and 
semantic similarities. Neither appending an R-Tree to each keyword with a huge 
space redundancy, nor just clustering all keywords into $k$ groups with a high false positive
ratio during query search, QC-Tree smartly splits keyword set into hierarchical
levels and link them by a Quad-Tree. 

To improve the searching efficiency, we propose a new hierarchical quad clustering algorithm based on the \emph{kernel $k$-means}~\cite{Dhillon2004Kernel}.
Compared with the traditional $k$-means clustering, kernel $k$-means will have better clustering effect even the samples do not obey the normal distribution and is more suitable to cluster the keywords.
Moreover, different from the common clustering, hierarchical clustering can form a meaningful relationship between different clusters, which is helpful to allocate a new sample and decrease the cost of misallocation.
After the clustering process finishs, a quad-cluster tree (QC-Tree) is used to arrange all of these clusters, which is the core composition of the keyword cluster layer. In Alg.~\ref{alg:cluster}, the critical part is applying the kernel $k$-means to each keyword cluster per level, with $k$ fixed as $4$. 
Furthermore, when the diameter of the keyword cluster is smaller than the $\tau_{cluster}$, the duplication operation is executed, which is presented in Alg. \ref{alg:duplication} and will be discussed later.

\begin{algorithm}[H]
	\SetKwFunction{Duplication}{Duplication}
	\SetKwFunction{KERNEL}{Kernel$k$Means}
	\caption{Hierarchical quad clustering algorithm}
	\KwIn{keyword set $K$, cluster number $k$}
	\KwOut{Quad-Cluster Tree: $T_{qc}$}
	\label{alg:cluster}
	$T_{qc}$.add($K$)\\
	Insert $K$ into a priority queue $U$ {\color{ForestGreen}\tcc*[f]{instert as a set}}\\
	\While {U $\neq$ $\emptyset$}
	{
		$S$ $\leftarrow$ $U$.Pop() {\color{ForestGreen}\tcc*[f]{pop the whole set }}\\
		\{$S_1$,$S_2$,$S_3$,$S_4$\} $\leftarrow$ \KERNEL($k$, $S$) {\color{ForestGreen}\tcc*[f]{$k$=4 by default}}\\
		{
			\ForEach{$S_i$ $\in$ \{$S_1, S_2,S_3, S_4$\}}
			{
				\eIf{$S_i$.diameter $ < $ $\tau_{cluster}$}
				{
					\Duplication($S_1$,$S_2$,$S_3$,$S_4$)
				}
				{
					insert  $S_i$ in to $U$
					
					$T_{qc}$.add($S_i$){\color{ForestGreen}\tcc*[f]{$S_i$ are children of $S$ }}
					
				}
			}
		}
		
		
	}
\end{algorithm}

\begin{algorithm}[H]
	\caption{Duplication}
	\KwIn{Four keyword sets: $S_1$,$S_2$,$S_3$,$S_4$ }
	\KwOut{Duplicated keyword sets: $S_1'$,$S_2'$,$S_3'$,$S_4'$}
	\label{alg:duplication}
	\For{$\forall k_i \in$ $S_1 \bigcup S_2 \bigcup S_3 \bigcup S_4$ }
	{
		\If({\tcc*[f]{\color{ForestGreen}Variance}}){$\sigma( d(k_i, S_j.cen$)) < $\tau_{dup}$  }
		{
			$S_j' \leftarrow$ $k_i \bigcup$  $S_j$, if $k_i \not\in S_j$ with $j \in \{1,2,3,4\}$
		}
	}
	\{$S_1, S_2, S_3, S_4$\} $\leftarrow$ \{$S_1', S_2', S_3', S_4' $ \}
\end{algorithm}

Fig. \ref{fig:hier:a} illustrates the hierarchical clustering in Alg.~\ref{alg:cluster}, where each dot represents a keyword and different aggregation of these dots presents different keyword clusters. The dots marked in different color are the centroid of these clusters, and moreover, same color denotes their clusters stay in the same level.

\begin{figure}
	\centering
	\subfigure[Distribution]{
		\label{fig:hier:a} 
		\includegraphics[width=1.27in]{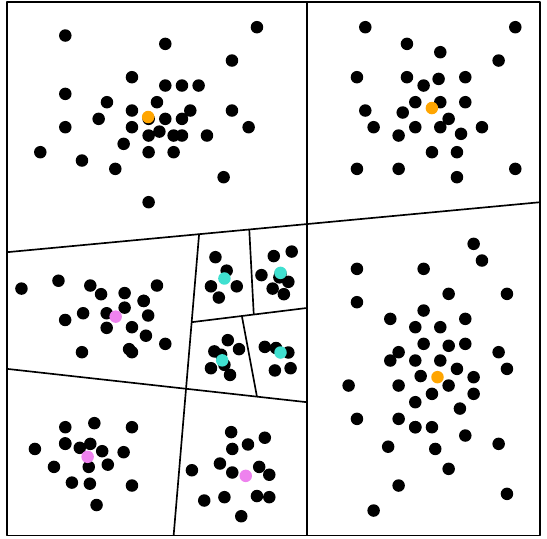}}
	\hspace{0.1in}
	\subfigure[Quad-cluster Tree]{
		\label{fig:hier:b} 
		\includegraphics[width=1.42in]{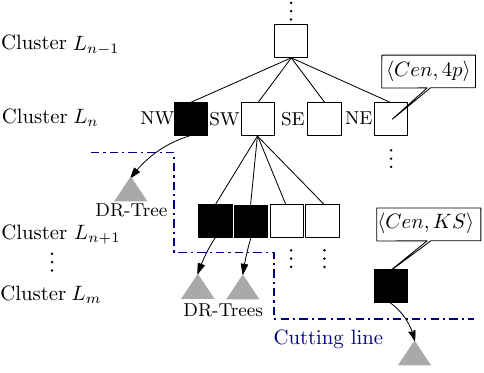}}
	\hspace{0.1in}
	\subfigure[Keyword relaxation]{
		\label{fig:hier:c} 
		\includegraphics[width=1.35in]{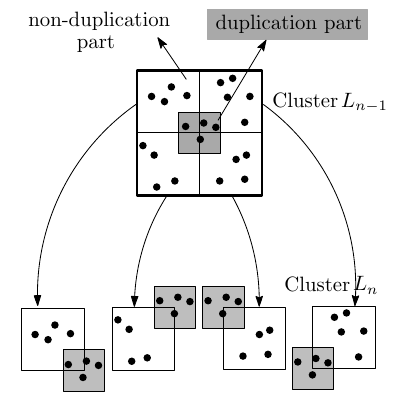}}
	\caption{Overview of the keyword cluster layer}
	\label{fig:hier} 
\end{figure}

Notice that, the main target of QC-Tree is to improve the pruning effect of keywords while making the future query keyword set located in only one keyword cluster.
As is shown in both Alg. \ref{alg:cluster} and Fig. \ref{fig:hier:a}, with the cluster level growing, the cluster will be more centralized and compact.
That means the possibility of  one query being allocated to different clusters increases layer by layer.
It is necessary to decide an optimal $\tau_{cluster}$ to terminate the hierarchical cluster proceeding, if not, there would  only be a single keyword in each cluster finally.
The basic structure of QC-Tree is displayed in Fig. \ref{fig:hier:b}, where each internal node keeps the centroid keyword ($cen$) and four pointers (4p) to its four descendants nodes,
and each leaf node will keep the keyword set in this cluster  and the pointer to a new DR-Tree. Additionally, a cut-line is drawn to emphasize the shift of index structure, which is mainly dependent on the value of $\tau_{cluster}$.

As is analyzed above, the leaf cluster is where a query would  most likely be scattered into different clusters. We will take a keyword-relaxation operation by duplicating some keywords among the four clusters sharing the same parent node.
In Fig. \ref{fig:hier:c}, for a keyword cluster, its keywords are grouped into four sub-clusters and the duplication operation need to be executed. The dots in the shadow represent the keywords that will be duplicated and allocated to all of these four sub-clusters because they are closed to all of the sub-clusters. Here, we introduce another threshold ($\tau_{dup}$) to decide whether to execute the duplication operation. 
Although this keyword-relaxation operator will cause redundancy of keywords and extra space consumption, it will largely improve the time efficiency, which will also be demonstrated in the experimental verification.

\subsection{Spatial layer}
Under each keyword cluster in the bottom of QC-Tree, we build a DR-Tree based on dual-filtering technique to organize the spatial objects in this cluster.

In Fig. \ref{fig:QDRTree}, a basic structure of DR-Tree  is shown in the spatial layer.
Each internal node $N$  records  a two-element tuple: $\langle SP, KB \rangle$. The first element $SP$ stands for the
skyline points of the numerical attributes of all objects in the subtree rooted
at the node. The  second element is a bitmap of the keywords included in this cluster, which uses 1 and 0 to denote the existence of keywords.

\subsubsection{Keyword bitmap filter algorithm:}\label{sct:kb2}
In the DR-Tree, each node just records the keyword bitmap, and then the specific keywords list is kept only in the leaf keyword cluster. Then, the keyword relevance can be calculated just by Bitwise AND within the pair of bitmaps, which can decrease the storage consumption and increase the query efficiency.

Because bitwise AND within bitmaps need an exact keywords matching, in order to support
similar keyword matching, we also implement the relaxation in each query process.
In Fig. \ref{fig:QDRTree}, as is highlighted in blue, the bitmap of query keywords
performs a search-relaxation by switching some 0-bits to 1-bits based on the keyword
similarity
The search-relaxation algorithm will be proposed in Alg. \ref{alg:search} in Section \ref{sct:kb2}.

\begin{figure}[h]
	\centering
	\includegraphics[width=0.78\textwidth]{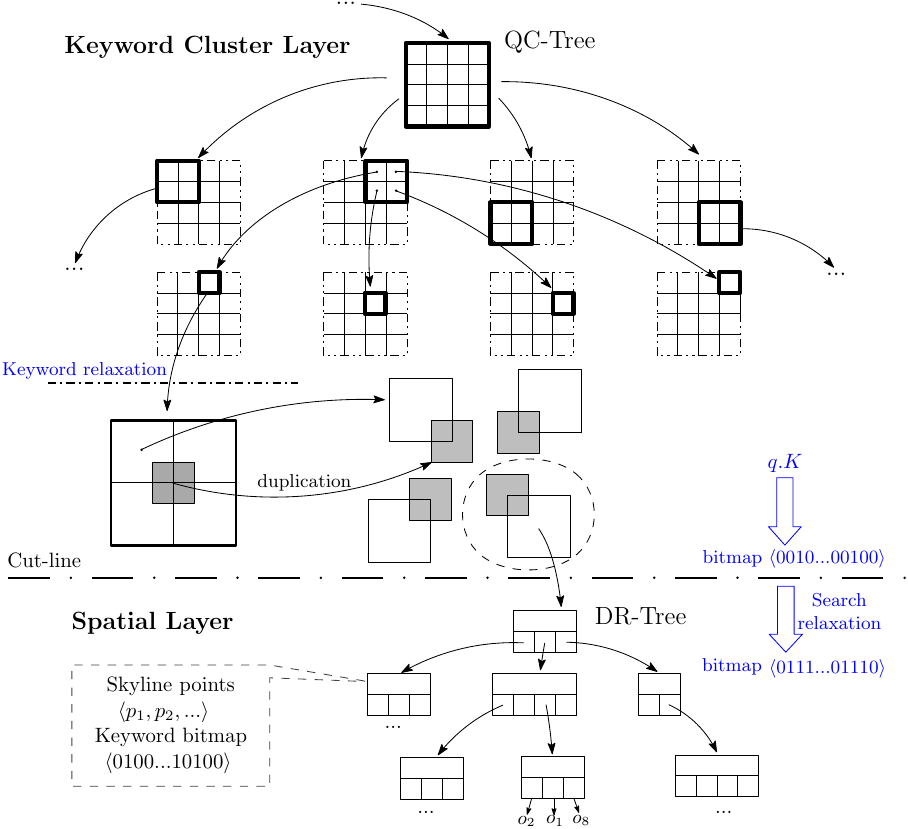}
	\caption{Structure of QDR-Tree}\label{fig:QDRTree}
\end{figure}

\subsubsection{Multidimensional subspace skyline filter algorithm: }
In order to satisfy the needs of user's
intention on multiple attributes,
a filter called Multidimensional Subspace Skyline
Filter, which is inspired by \cite{lee2014toward,borzsony2001the}, is employed to amortize the query false positive and the cost of computation.
We use the Evaluate() algorithm proposed in \cite{lee2014toward} to gain the multidimensional
skyline points efficiently, and then let every QC-Tree node record the skyline points of its
descendants.
Furthermore, in order to reduce the complexity of recording multidimensional skyline points, we
will take the point-compression operation by merging the closed skyline points
in the attributes space.
We calculate the cosine distance between skyline points' attributes
to measure the similarity, and then merge these closed points when cosine distance is larger than a threshold.

\section{QDR-based Query Algorithm} \label{enhanced}
In this section, we will introduce the ASKQ processing algorithms based on QDR-Tree. The process includes finding the Leaf Cluster, making search-relaxation and searching in the DR-Tree.

\textbf {Find the leaf cluster}
The leaf keyword cluster that is best-matched with $q$ can be obtained by iteratively comparing $q$ with the four sub-clusters in each cluster level.
If the combination of keywords in the query is typical and can be allocated into the same cluster, only one keyword cluster will be found. Otherwise, more than one keyword cluster may be returned.

\textbf {Search-relaxation }\label{sct:kb}
As is stated in Section \ref{sct:kb2}, by means of executing search-relaxation,
bitmap-based filter can support similar keyword matching. In Alg.
~\ref{alg:search}, a bitmap of relaxed query keyword is obtained by
switching 0-bit to 1-bit if their keyword distance is under a threshold.
By adopting a rational threshold, we can make a good trade-off between time cost and space occupation.

\begin{algorithm}[H]
	\caption{FindLeafCluster}
	\label{alg:find}
	\KwIn{$q$, QC-Tree $T_{qc}$}
	\KwOut{the leaf cluster: $LC$}
	$LC$ $\leftarrow$ $\emptyset$\\
	\ForEach {k $\in$ q.K}
	{
		$lc$ $\leftarrow$ $T_{qc}.root$ \\
		\While{lc is not leaf cluster}
		{
			$ls$ $\leftarrow$ $lc$.$sub_i$, with $d(k, lc.sub_i.cen)$ is minimum among  4 $lc.subs$\\
		}
		$LC$ $\leftarrow$ $LC \cup lc$ 
		
	}
\end{algorithm}

\begin{algorithm}[H]
	\caption{Search relaxation}
	\label{alg:search}
	\KwIn{bitmap of query keyword: $bmq$, bitmap of keyword cluster: $bmc$ }
	\KwOut{bitmap of relaxed query keyword: $bmr$}
	\For {i $\leftarrow$ 1 to |bmq|}
	{
		\If{$bmq[i]$ = 1}
		{$bmr[i]$ $\leftarrow $ 1\\
			{\For {j $\leftarrow$ 1 to |bmc|}
				{\If{d($k_i$,$k_j$) $<$ $\tau$}
					{$bmr[j]$ $\leftarrow $ 1\\}}}
	}}
\end{algorithm}

Alg. \ref{alg:qdr_search} illustrates the query processing mechanism over QDR-Tree.
Given a query $q$, the object retrieval is carried out firstly by traversing the QC-Tree to locate the best-matched keyword cluster.
Secondly, after executing search-relaxation, it will traverse the DR-Tree
in the ascending order of the scores and keep a minimum heap for the scores.
Notice that, if more than one keyword cluster is located, it will
traverse all of them.
At last the
Top-$\kappa$ results can be returned.

%

The ranking score of an object $o$ for ASKQ is calculated by Eqn.~\eqref{eq:score}. Here, $\alpha,\beta \in [0,1]$ are parameters indicating the relative
importance of these three factors. $\psi(q,o)$ is the Euclidian distance between
$q$ and $o$. The $D^{max}_s$ is the maximal spatial distance that the client will
accept.
$\phi(q,o)$ which represents the keyword relevance between $q$ and $o$ is determined by the result
of Bitwise AND between their keyword bitmaps.
The smaller the score, the higher the
relevance.



\begin{equation}
score(q,o) = \alpha\beta \times \frac{\psi(q,o)}{D^{max}_s}
+ (1-\beta)\times \frac{1}{\phi(q,o)}
+(1-\alpha)\beta \times \sum_{i=1}^{|q.W|} q.w_i \times o.a_i
\label{eq:score}
\end{equation}

What is more,
the score for non-leaf node $N$ can also been measured to represent the optimal score of its 
descendant nodes, which is defined as Eqn.~\eqref{eq:score2}

\begin{equation}
\begin{split}
score(q,N) = &\alpha\beta \times \frac{\min \psi(q,N.MBR)}{D^{max}_s}
+ (1-\beta)\times \frac{1}{\phi(q,N)} \\
& +(1-\alpha)\beta \times\min_{\forall p \in N.sp} \sum_{i=1}^{|q.W|} q.w_i \times p.a_i
\end{split}
\label{eq:score2}
\end{equation}

where the $\min \psi(q, N.MBR)$ represents the minimum Euclidian distance between the N's MBR and the $\phi(q, N)$ is can also be calculated by the bitmap of keywords kept in this node.
We can prove that Top$_\kappa(q)$ is an exact result by the  Theorem \ref{the:1}. If the score of the internal node dose not satisfy the ASKQ, there is no need to search its descendant nodes. Hence, the final Top-$\kappa$ objects will have the least $\kappa$ scores.

\begin{theorem}\label{the:1}
	The score of an internal node N is the best score of its descendant object $o$ to the query $q$.
\end{theorem}
\begin{proof}
	the score factors in location proximity, keyword relevance and non-spatial attributes' value.
	First, the MBR of the N encloses all of its descendant objects, then $\forall o_i \in$ descendant objects of N, $\min \psi(q, N.MBR) \leq \psi(q, o_i)$. Second, the keyword bitmap includes all of the keywords existing in the descendant objects of $N$. Obviously, $\phi(q, N) \geq \phi(q, o)$.   Finally, the skyline points dominate or are equal to all of descendent objects concerning the value of attributes,
	i.e., $\min_{\forall p \in N.SP} \sum_{i=1}^{|q.W|} q.w_i \times p.a_i \leq \sum_{i=1}^{|q.W|} q.w_i \times o.a_i$. All these inequalities contribute to that $score(q, N) \leq score(q.o)$.
	\hfill $\Box$
\end{proof} 

\begin{algorithm}[H]
	\SetKwFunction{FindLeafCluster}{FindLeafCluster}
	\SetKwFunction{SearchRelaxation}{SearchRelaxation}
	\caption{QDR-Search algorithm}
	\label{alg:qdr_search}
	\KwIn{a query $q$, Top$_\kappa$ results $\kappa$, and a QDR-Tree $T_{qdr}$ }
	\KwOut{Top$_\kappa$($q$)}
	$LC$ = \FindLeafCluster($q$, $T_{qc}$);\\
	\For{i $\leftarrow$ 1 to $|LC|$}
	{
		$q$.bitmap $\leftarrow$ \SearchRelaxation($q$.bitmap, $LC[i]$.bitmap)\\
		Minheap.insert($LC[i]$.root, 0)\\
		\While {Minheap.size() $\neq$ $0$}
		{
			$N$ $\leftarrow$ Minheap.first() \\
			\eIf {$N$ is an object}
			{Top$_\kappa(q)$.insert($N$) \\
				\If{Top$_\kappa(q)$.size() $\geq$ $k$}
				{\bf break}
			}
			{
				\For {$n_i$ $\in$ $N$.entry}
				{
					\If{Number of objects with smaller score than score($q$, $n_i$)  in Minheap  $<$ ($\kappa$ $-$ Top$_\kappa$($q$).size())}
					{Minheap.insert($n_i$, score($q$, $n_i$))\\}
				}
			}
	}}
\end{algorithm}

\section{Experiment Study}\label{performance}

\subsection{Baseline algorithm}\label{baseline}
In this section, we propose three baseline algorithms which are based on the three existing indexes listed in Tab. \ref{tab:1}, including IR-Tree\cite{li2011ir}, SKY R-Tree \cite{sasaki2014sky} and IRS-Tree\cite{6940236}. As is discussed in Sec. \ref{sec-intr}, none of these existing indexes can be qualified for the ASKQ due to different drawbacks. The specific algorithm designs will be respectively explained in detail as follows.

Because the IR-Tree pays no attention on the value of numerical attributes, all spatial objects containing the query keywords and
numerical attributes will be extracted. After that they will be ranked by the comprehensive value of numercial attributes. Eventually, the top-$\kappa$ spatial objects are just the result of the ASKQ.

Different from the IR-Tree, the SKY R-Tree fails to support multi-keywords query because one SKY R-Tree
can only arrange one keyword and its corresponding spatial objects, such as restaurant.
In order to deal with the ASKQ, all of the SKY R-Trees containing the query keywords will be searched and merged to obtain the final top-$\kappa$ results.

The last baseline algorithm is proposed based on the IRS-Tree which is originally intended
to address the GLPQ. 
Unlike ASKQ, the GLPQ requires specific range of attributes to leverage the IRS-Tree.  
To copy with the ASKQ, we will firstly set some different suitable ranges of each attributes as the input, which insures that enough spatial objects can be returned. Afterwards, we will further to select top-$\kappa$ objects from the results in the first stage.
Apparently, in our experiments, the IRS-Tree will not make much sense anymore.

Notice that, all of these three baseline algorithms cannot solve the ASKQ directly at a time and need subsequent elimination of redundancy, which determines their inefficiency in the ASKQ.

In the experiment section, we conduct extensive experiments on both real and synthetic datasets to evaluate
the performance of our proposed algorithms.

\subsection{Experiment Setup}
The real dataset is crawled from the famous location-based service platform, Foursquare.
After information cleaning, the dataset has about 1M objects consisting of geographical location,
the keyword list written in English, and the normalized value of attributes.
Each spatial object contains the keywords such as steak,
pizza, coffee, etc. and four numerical attributes, including price, environment, service and rating.

In the synthetic dataset, each object is composed of coordinates, various keyword, and multi-dimensional
numerical attributes. The size of the synthetic dataset
varies in the experiments.
The
coordinates are randomly generated in (0, 10000.0),
and the average number of keywords per object is decided by a parameter $r$ which denotes the
ratio of the number of object's keywords to the cluster's.
Without loss of generality, the values of each numercial attribute are randomly and independently generated,
following a normal distribution.

\begin{table}
	\caption{Default value of parameters}
	\begin{center}
		\begin{tabular}{ccc}
			\hline
			Parameter & Default value & Descriptions \\
			\hline
			$\kappa$ & 10 & Top-$\kappa$ query \\
			$|o.A|$ & 4 & No. of attributes' dimension \\
			$\delta$ & 0.5 & Weight factor of Eqn.~\eqref{eq:key}\\
			$\alpha$ & 0.5 & Weight factor of Eqn.~\eqref{eq:score}\\
			$\beta$ & 0.67 & Weight factor of Eqn.~\eqref{eq:score}\\
			$\tau_{cluster}$ & 0.3 & Threshold of quad clustering\\
			$\tau_{dup}$ & 0.05 & Threshold of duplication\\
			$|O|$ & 1M & Number of objects \\
			$M$ & 25 & Maximum number of DR-Tree entries \\
			\hline
		\end{tabular}
	\end{center}
	\label{valuetable}
\end{table}

We compare the query cost of proposed algorithms with different datasets respectively.
The experimental settings are given in Table \ref{valuetable}. The default values
are used unless otherwise specified. All algorithms are implemented in Python and run with Intel core i7 6700HQ CPU at 2.60 GHz and 16 GB memory.

\subsection{Performance Evaluation}
In this section, we campare
different baseline algorithms proposed in Section \ref{baseline} with our framework.
We evaluate the processing time and disk I/O 
of all the proposed methods by varying
the parameters in Table \ref{valuetable} and investigate their effects.
In the first part we study the experimental results on the real dataset.

{\bf Index construction cost: }
We first evaluate the construction costs of various methods. The cost of an index is measured
by its construction time and space budget. The costs of various methods are shown in Fig. \ref{fig:indexcost},
where IRS refers to the baseline structure IRS-Tree, the SKY-R, IR refer to SKY R-Tree, 
IR-Tree respectively, and the QDR represents our design. We can see that, SKY R-Tree and IR-Tree exceed both in time consumption and space cost, because they are short of
attention of either attributes or multi-keywords. Moreover, since we employ bitmap
and skyline points to measure numercial attributes, QDR is more lightweight than IRS.

{\bf Effect of $\kappa$: }
We investigate the effect of $\kappa$ on the processing time and disk I/O of the proposed algorithms by randomly
generate 100 queries. Here, considering that the SKY R-Tree and IR-Tree do not take into account either 
attributes or keywords, we add a filter operation after their query process. For example, the SKY R-Tree returns
Top-$\kappa$ results of each keyword and merges them in the second stage. Obviously,
this redundancy of result is the main reason of the high time cost.    
As shown in Fig. \ref{fig:effk}, with the increase of $\kappa$, IRS and QDR have the
same smoothly increasing trend on query time and disk I/O. QDR exceeds in query time cost with different
parameters. It indicates that we can effectively receive the Top-$\kappa$ results from one branch to another.

\begin{figure}[!htbp]
	\begin{minipage}[t]{0.5\linewidth}
		\centering
		\subfigure[Index time]{
			\label{fig:indexcost:a} 
			\includegraphics[width=0.45\textwidth]{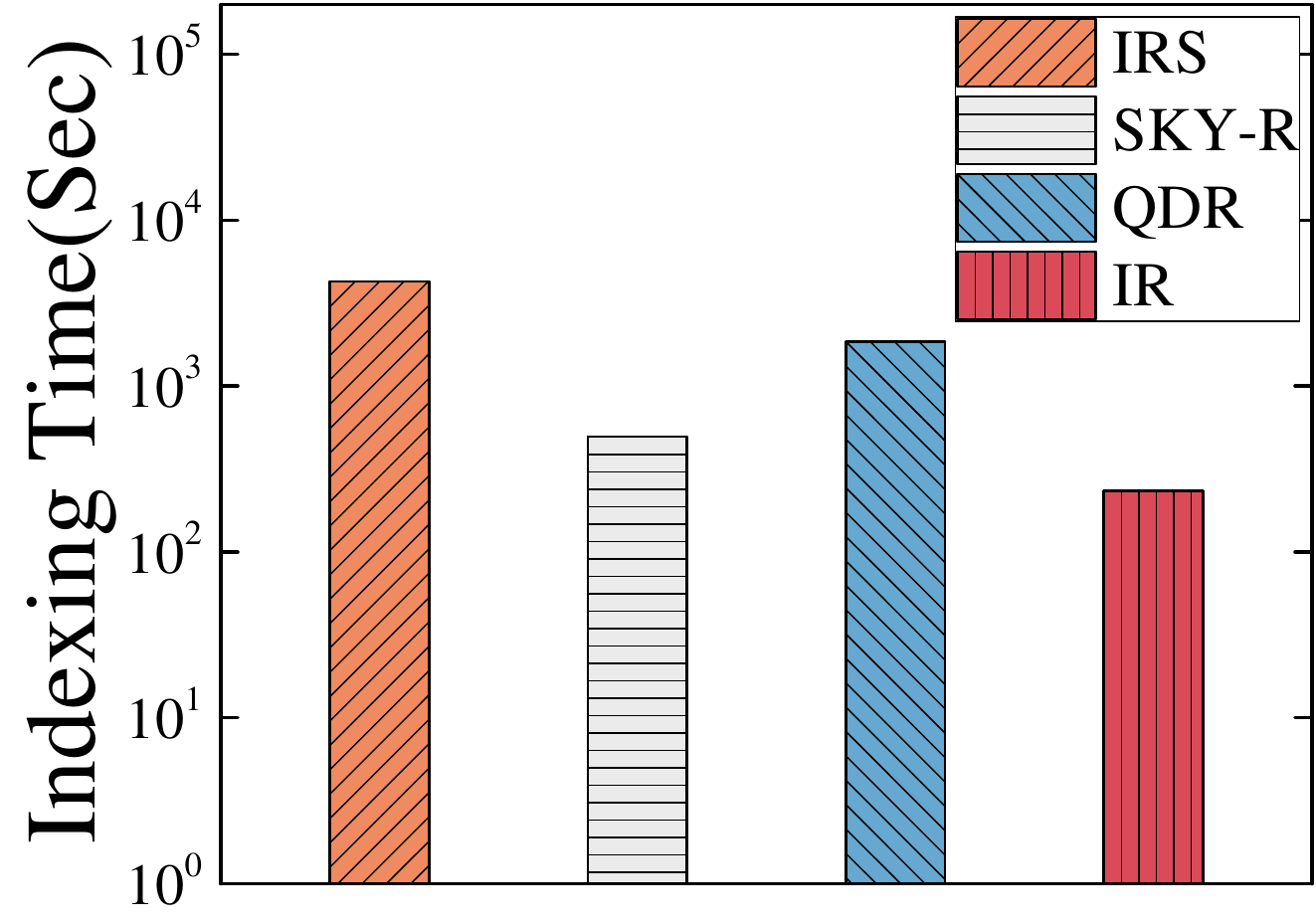}}
		\subfigure[Index size]{
			\label{fig:indexcost:b} 
			\includegraphics[width=0.45\textwidth]{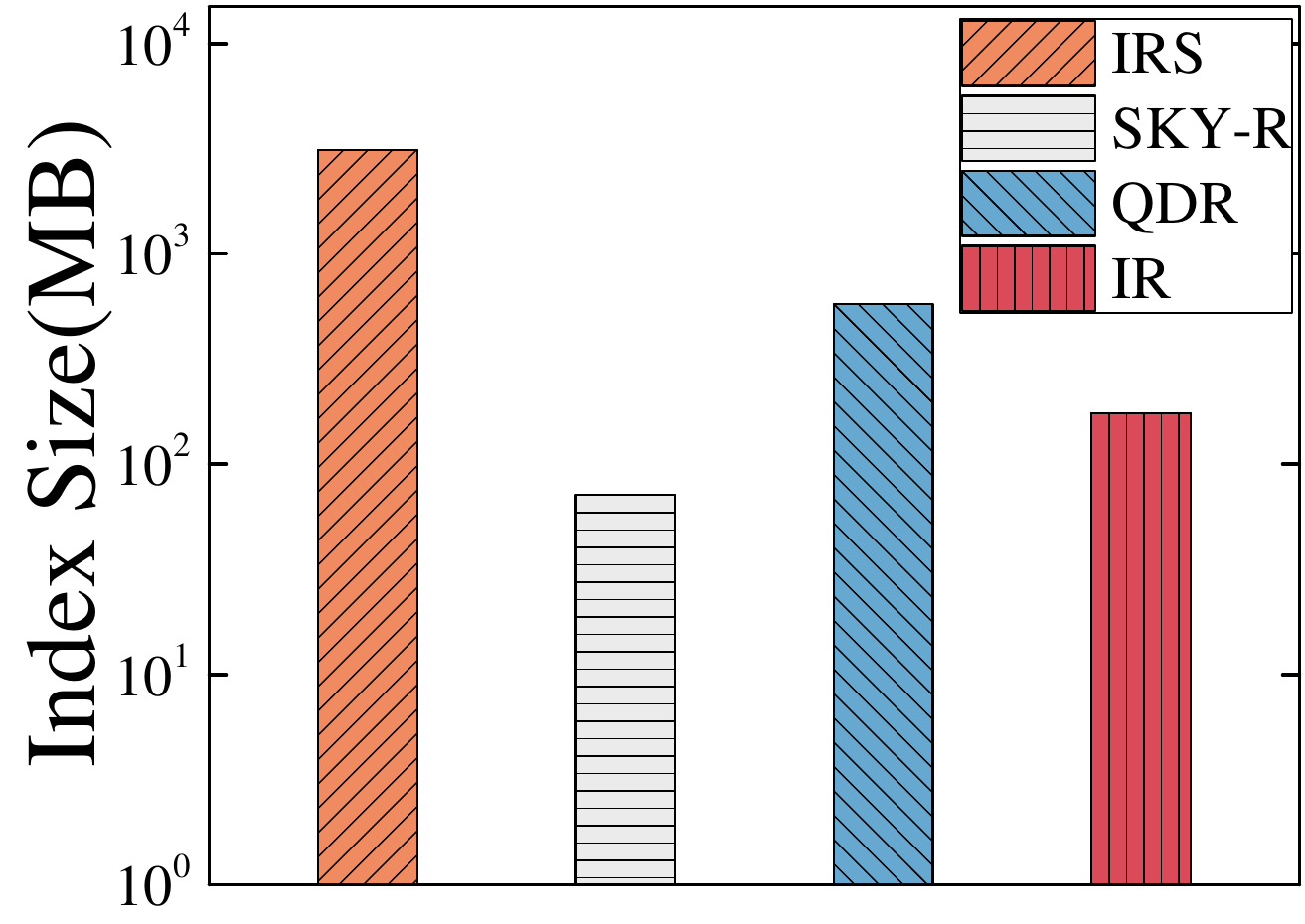}}
		\caption{Construction cost} \label{fig:indexcost}
	\end{minipage}
	\begin{minipage}[t]{0.5\linewidth}
		\centering
		
		\subfigure[Process time]{
			\label{fig:effk:a} 
			\includegraphics[width=0.43\textwidth]{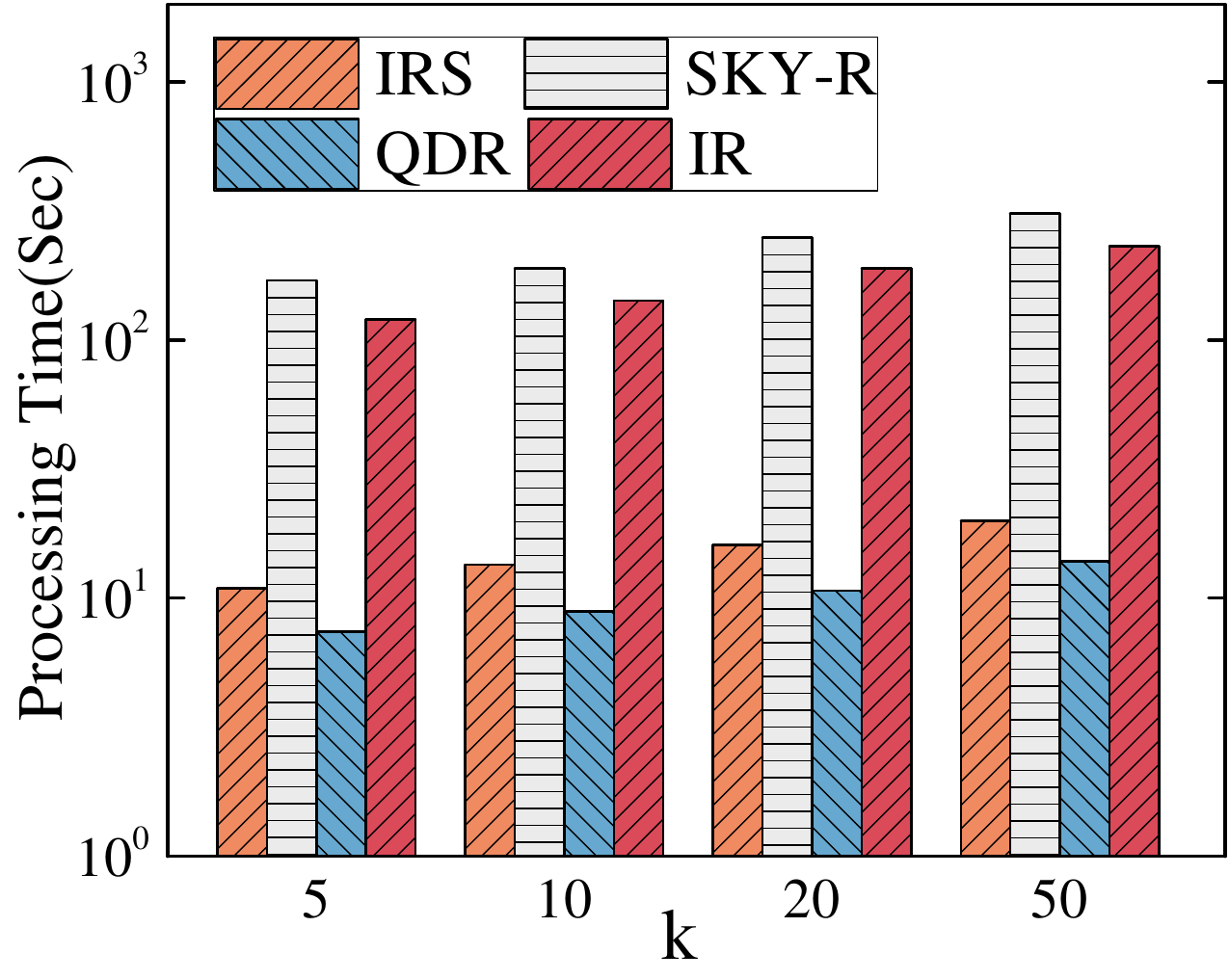}}
		\subfigure[Disk I/O]{
			\label{fig:effk:b} 
			\includegraphics[width=0.43\textwidth]{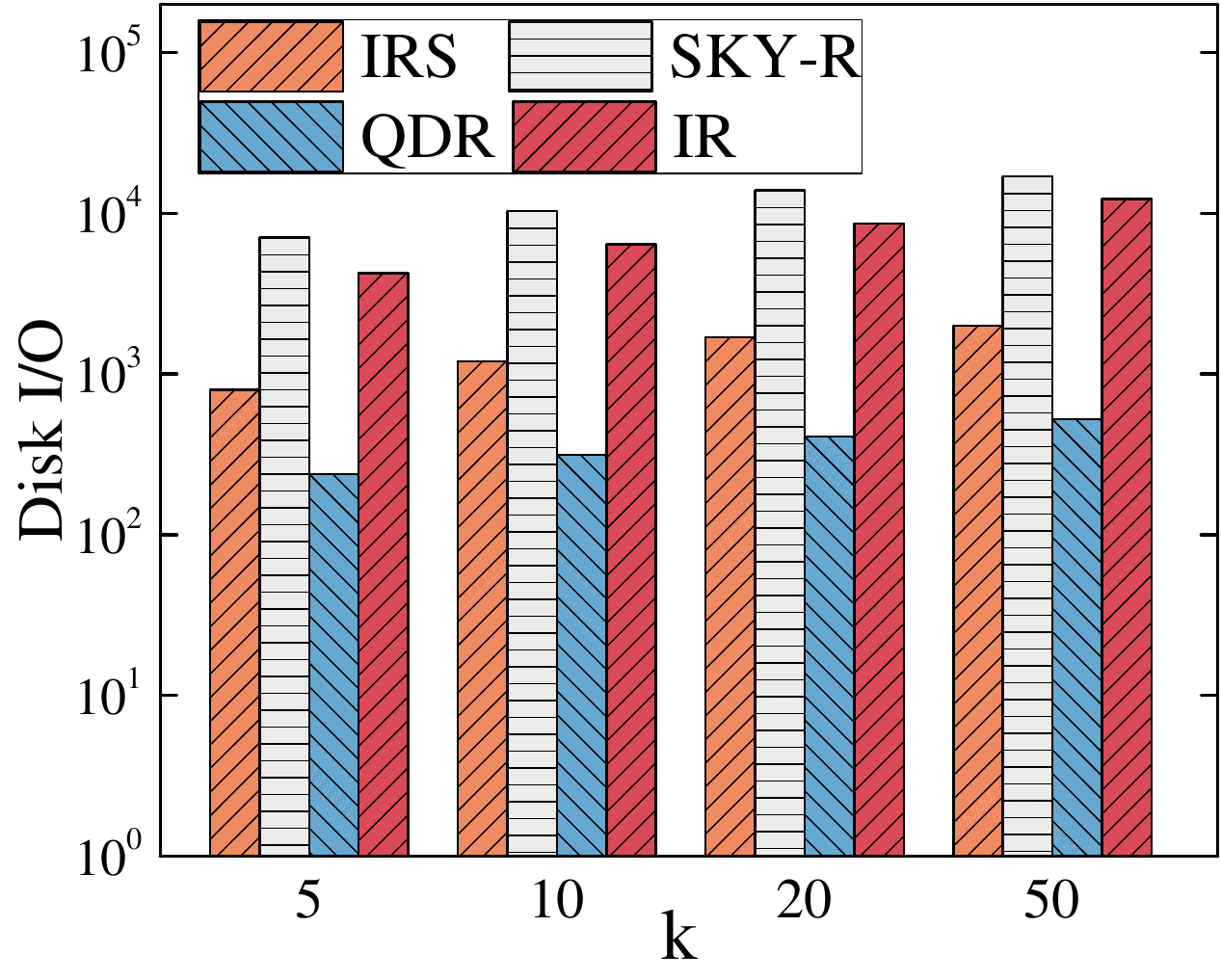}}
		\caption{Effect of $\kappa$}
		\label{fig:effk} 
		
	\end{minipage}
\end{figure}


{\bf Effect of $|O|$: }
Parameter $|O|$ denotes the number of objects in the QDR-Tree, We
increase the number of objects in the synthetic dataset from
0.1 to 5 M. It can be seen from Fig. 7(a)(b), that, SKY-R and IR have more obvious
increases in query time cost and disk I/O when the data size
increases, which can be explained because of their larger redundancy along with the larger dataset.
On the other hand, the QDR is more stable and surpasses another three indexes.

{\bf Effect of $|o.A|$: }
Parameter $|o.A|$ denotes the number of attributes the object $o$ covers. 
As shown in Fig. 7(c)(d),
the query time and Disk I/O of the IR-Tree based on synopses tree has distinct increase, because it fails to consider the attributes,  while another three frameworks
are more stable, which is mainly because of either the skyline filter algorithm or the synopses tree.

{\bf Effect of $\tau_{cluster}$ \& $\tau_{dup}$: } $\tau_{cluster}$ and $\tau_{dup}$ are the crucial parameters in our QDR-Tree, which are analyzed theoretically in Sec. \ref{SCIR}. The experimental results also verify their effect on the Processing Time and Index Size.
Fig. 6(a)(c) show that both $\tau_{cluster}$ and $\tau_{dup}$ have an optimal value to minimize the processing time. Smaller or larger value will both increase the processing time because of keyword scattring or redundancy. In Fig. 6(b)(d), index size decreases as the $\tau_{cluster}$ becomes larger and reach saturation at some point, while it increases along with the $\tau_{dup}$ because of the increase of keyword redundancy. 

\begin{figure}
	\centering
	\subfigure[Efficiency-$|O|$]{
		\label{fig:effoa:a} 
		\includegraphics[width=1.1in]{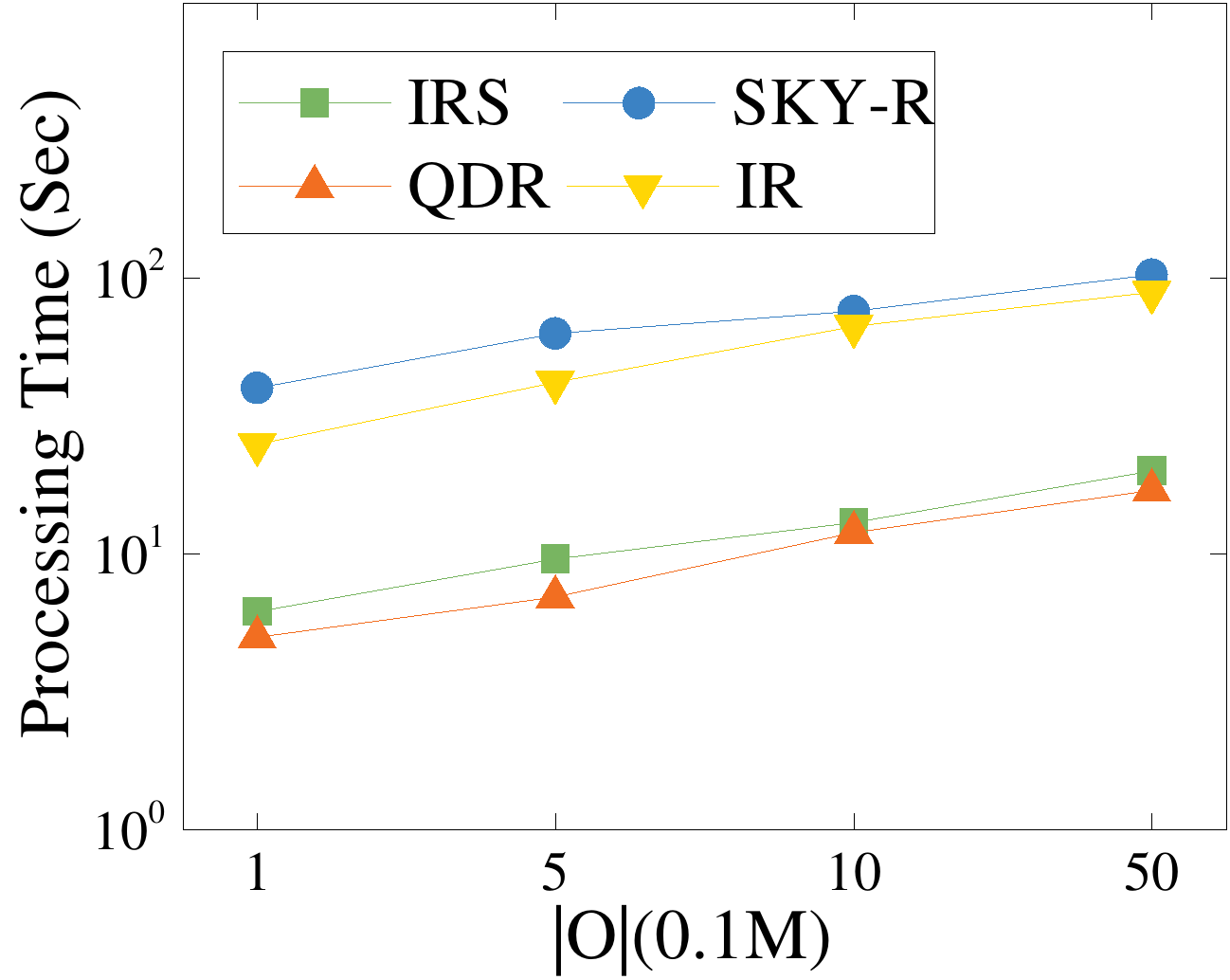}}
	\subfigure[Disk I/O-$|O|$]{
		\label{fig:effoa:b} 
		\includegraphics[width=1.1in]{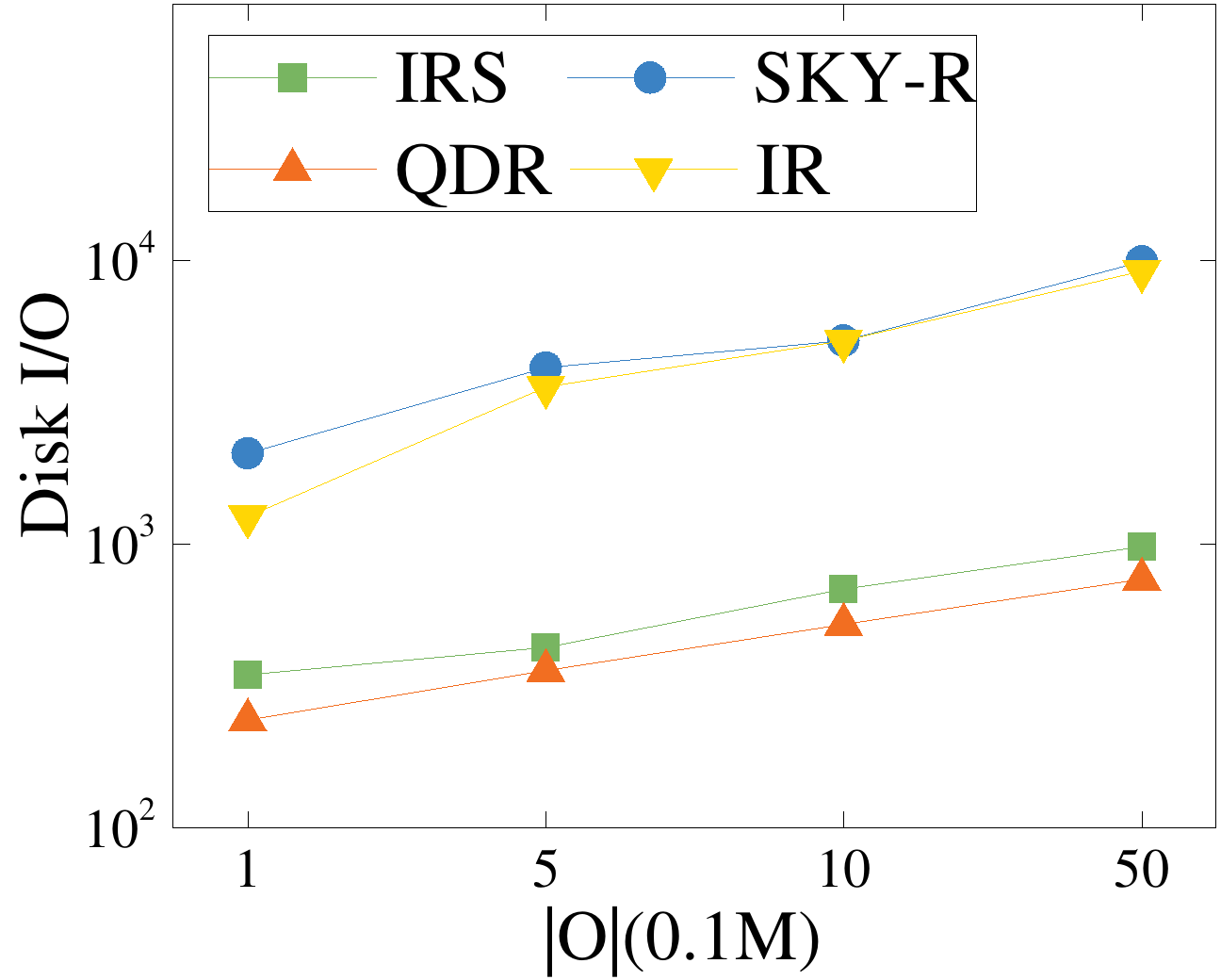}}
	\subfigure[Efficiency-$|o.A|$]{
		\label{fig:effoa:c} 
		\includegraphics[width=1.07in]{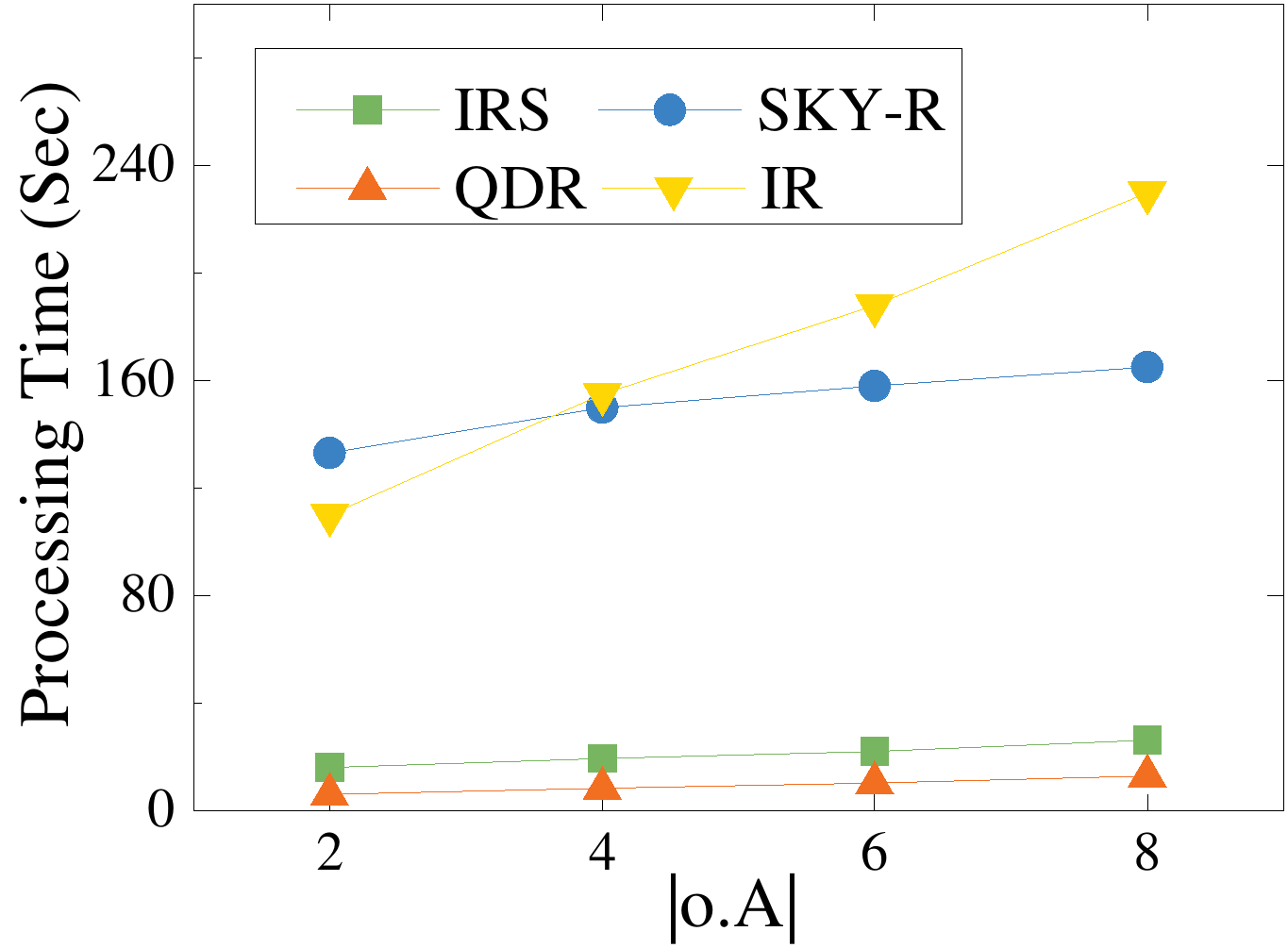}}
	\subfigure[Disk I/O-$|o.A|$]{
		\label{fig:effoa:d} 
		\includegraphics[width=1.07in]{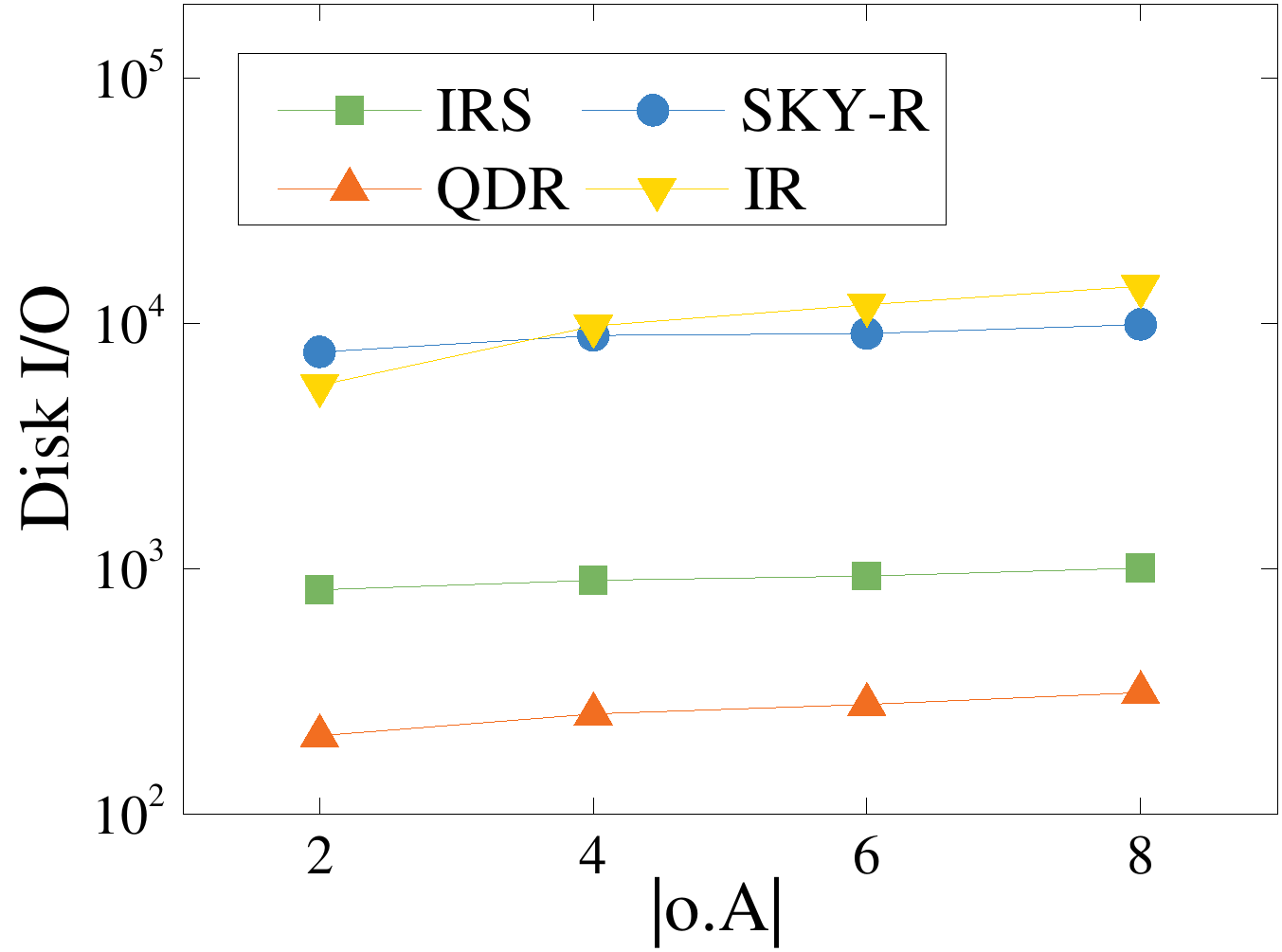}}
	\caption{Synthetic dataset}
	\label{fig:effoa} 
\end{figure}

\begin{figure}
	\centering
	\subfigure[\scriptsize Efficiency-$\tau_{cluster}$]{
		\label{fig:effab:a} 
		\includegraphics[width=1.05in]{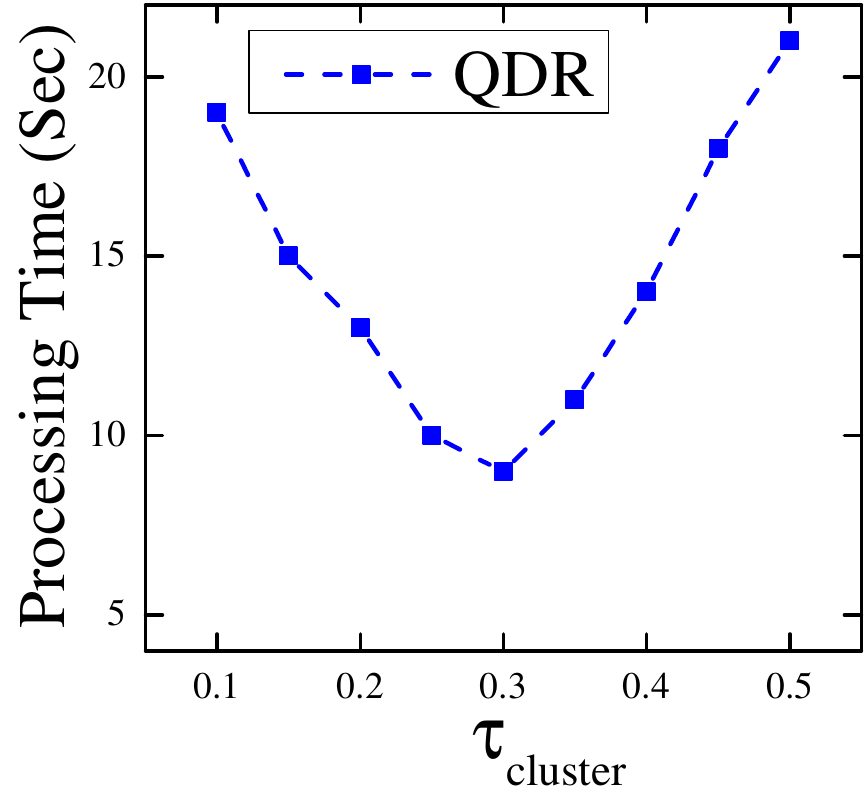}}
	\subfigure[\scriptsize Index Size-$\tau_{cluster}$]{
		\label{fig:effab:b} 
		\includegraphics[width=1.1in]{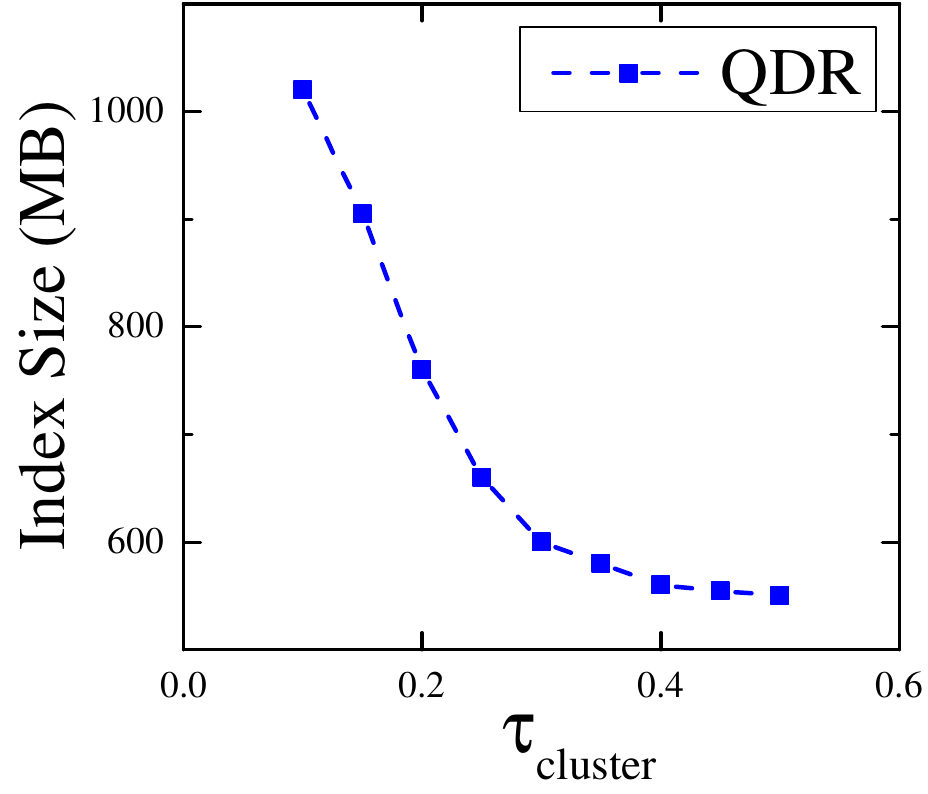}}
	\subfigure[\scriptsize Efficiency-$\tau_{dup}$]{
		\label{fig:effab:c} 
		\includegraphics[width=1.1in]{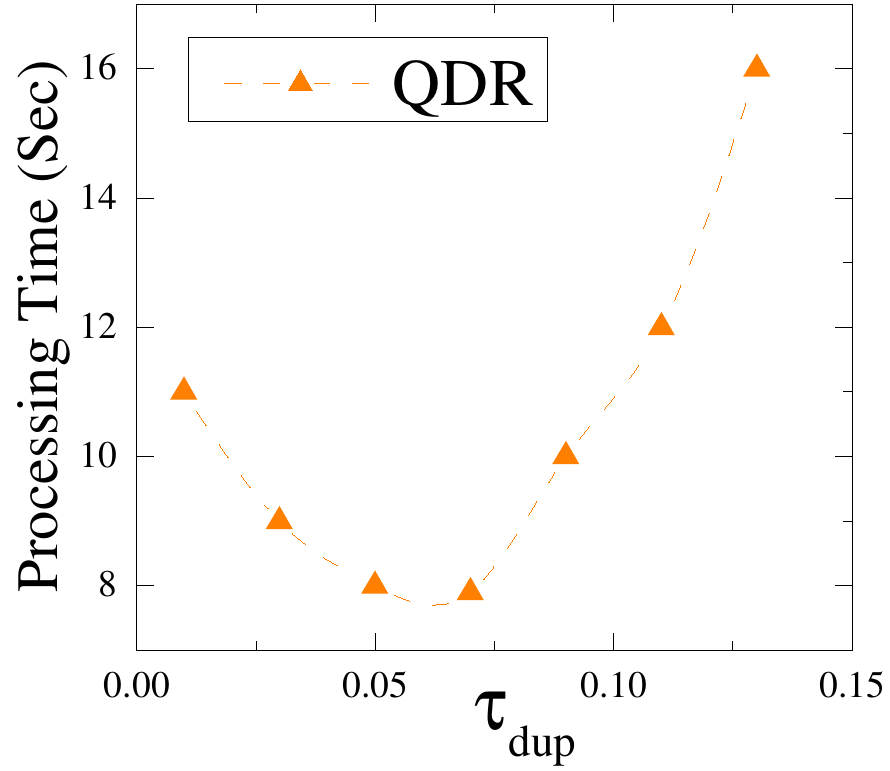}}
	\subfigure[\scriptsize Index Size-$\tau_{dup}$]{
		\label{fig:effab:d} 
		\includegraphics[width=1.1in]{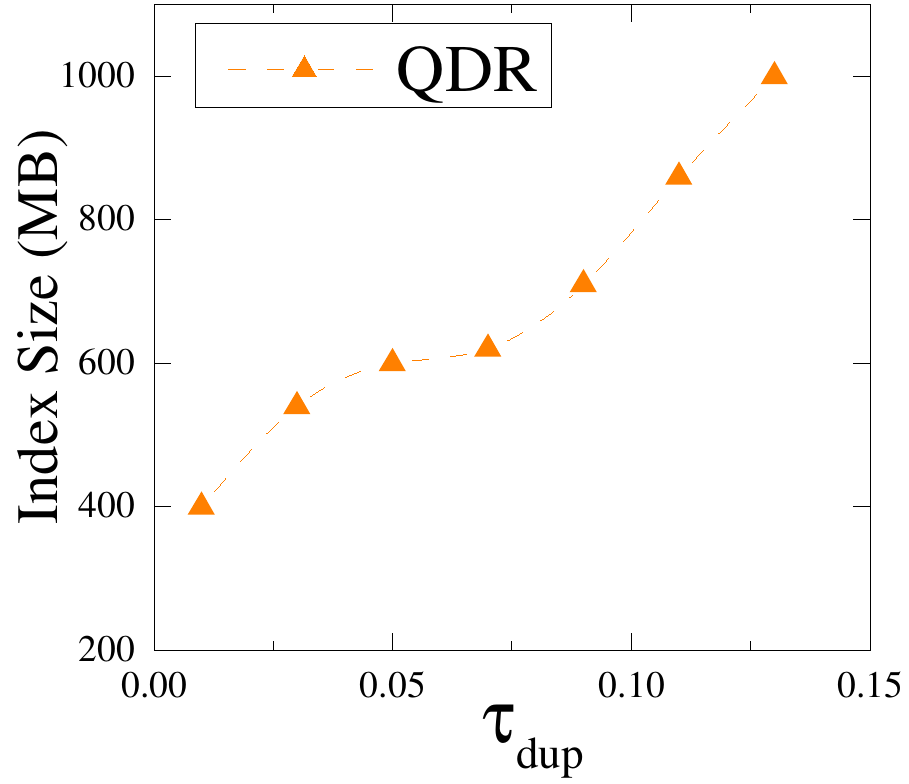}}
	\caption{Effect of $\tau_{cluster}$ \& $\tau_{dup}$}
	\label{fig:effab} 
\end{figure}


\section{Conclusion} \label{conclusion}
In this paper, we formulated the attributes-aware spatial keyword query (ASKQ) and
proposed a novel index structure call Quad-cluster Dual-filtering R-Tree (QDR).
QDR-Tree is a two-layer hybrid index based on two index structures and two searching
algorithms. We also proposed a novel method to measure the relevance of
spatial objects with query keywords, which applies keyword-bitmap and search-relaxation
to achieve exact and similar keyword match. Moreover, by employing the keyword-relaxation,
we greatly improve the time efficiency at the sacrifice of a little space consumption.
Finally, massive experiments
with real datasets demonstrate the efficiency and effectiveness of QDR-Tree.

\bibliographystyle{splncs03}
\bibliography{ref/reference}

\end{document}